\documentclass[10pt, twocolumn, a4paper]{article}
\pdfoutput=1
\input{packages.sty}
\input{commands.sty}
\usepackage{format}

\begin{document}
\title{\Large\bf Box Facets and Cut Facets of Lifted Multicut Polytopes}
\author[1]{Lucas Fabian Naumann}
\author[1]{Jannik Irmai}
\author[1,2]{Shengxian Zhao}
\author[1,2]{Bjoern Andres}
\affil[1]{Faculty of Computer Science, TU Dresden}
\affil[2]{Center for Scalable Data Analytics and AI, Dresden/Leipzig}
\affil[ ]{\textit{bjoern.andres@tu-dresden.de}}
\date{}

\maketitle

\begin{abstract}
The \emph{lifted multicut problem} is a combinatorial optimization problem whose feasible solutions relate one-to-one to the decompositions of a graph $G = (V, E)$.
Given an augmentation $\widehat{G} = (V, E \cup F)$ of $G$ and given costs $c \in \mathbb{R}^{E \cup F}$, the objective is to minimize the sum of those $c_{uw}$ with $uw \in E \cup F$ for which $u$ and $w$ are in distinct components.
For $F = \emptyset$, the problem specializes to the multicut problem, and for $E = \tbinom{V}{2}$ to the clique partitioning problem. 
We study a binary linear program formulation of the lifted multicut problem.
More specifically, we contribute to the analysis of the associated \emph{lifted multicut polytopes}:
Firstly, we establish a necessary, sufficient and efficiently decidable condition for a lower box inequality to define a facet.
Secondly, we show that deciding whether a cut inequality of the binary linear program defines a facet is \textsc{np}-hard.
\end{abstract}   

\section{Introduction}
\label{section:introduction}

The \emph{lifted multicut problem} \citep{keuper-2015a} is a combinatorial optimization problem whose feasible solutions relate one-to-one to the decompositions of a graph $G = (V, E)$.
A decomposition of a graph $G = (V, E)$ is a partition $\Pi$ of the node set $V$ such that for every $U \in \Pi$ the subgraph of $G$ induced by $U$ is connected.
Given an augmentation $\widehat{G} = (V, E \cup F)$ of $G$ with $F \cap E = \emptyset$ and costs $c \in \mathbb{R}^{E \cup F}$, the objective of the lifted multicut problem is to minimize the sum of those costs $c_{uw}$ with $uw \in E \cup F$ for which $u$ and $w$ are in distinct components.
For $F = \emptyset$, the lifted multicut problem specializes to the multicut problem \citep{deza-1992,chopra-1993,chopra-1995}.
For $E = \tbinom{V}{2}$, it is equivalent to the clique partitioning problem \citep{groetschel-1989}. 

We study the binary linear program formulation of the lifted multicut problem by \citet{hornakova-2017} in which variables $x \in \{0,1\}^{E \cup F}$ indicate for all $uw \in E \cup F$ whether the nodes $u$ and $w$ are in the same component, $x_{uw} = 0$, or distinct components, $x_{uw} = 1$:

\begin{definition}{\citep[Def.~9]{hornakova-2017}}
\label{definition:lifted-multicut-problem}
For any connected graph $G = (V, E)$, any augmentation $\widehat{G} = (V, E \cup F)$ with $F \cap E = \emptyset$, and any $c \in \mathbb{R}^{E \cup F}$, the instance of the \emph{(minimum cost) lifted multicut problem} has the form
\begin{align}
\min \quad \left\{
	\sum_{e \in E \cup F} c_e \, x_e 
	\; \middle| \;
	x \in X_{G\widehat{G}}
\right\}
\end{align}
with $X_{G\widehat{G}}$ the set of all $x \in \{0,1\}^{E \cup F}$ that satisfy the following linear inequalities that we discuss in \Cref{section:preliminaries}:

\begin{align}
\forall C \in \text{cycles}(G) \; \forall e \in E_C \colon \ 
& x_e \leq \sum_{e' \in E_C \setminus \{e\}} x_{e'}
\label{eq:cycle}
\\
\forall uw \in F \; \forall P \in uw\text{-paths}(G) \colon \ 
& x_{uw} \leq \sum_{e \in E_P} x_e
\label{eq:path}
\\
\forall uw \in F \; \forall \delta \in uw\text{-cuts}(G) \colon \ 
& 1 - x_{uw} \leq \sum_{e \in \delta} (1 - x_e)
\label{eq:cut}
\end{align}
\end{definition}

\begin{figure}
\centering
\begin{tabular}{@{}l@{}}
\begin{tikzpicture}
\node[style=vertex] at (0, 0) (a) {};
\node[style=vertex] at (1, 0) (b) {};
\node[style=vertex] at (0.5, 0.71) (c) {};
\draw (a) -- (b);
\draw[mygreen] (b) -- (c);
\draw (c) -- (a);
\node at (0.5, 0.15) {$e_1$};
\node at (0.05, 0.4) {$e_2$};
\node[mygreen] at (0.95, 0.4) {$f$};
\end{tikzpicture}
\end{tabular}
\hfill
\begin{tabular}{@{}l@{\ }l@{}}
\begin{tikzpicture}
\node[style=vertex] at (0, 0) (a) {};
\node[style=vertex] at (1, 0) (b) {};
\node[style=vertex] at (0.5, 0.71) (c) {};
\draw (a) -- (b);
\draw[mygreen] (b) -- (c);
\draw (c) -- (a);
\node at (0.5, 0.15) {0};
\node at (0.15, 0.45) {0};
\node at (0.85, 0.45) {\color{mygreen}0};
\end{tikzpicture}
& \begin{tikzpicture}
\node[style=vertex] at (0, 0) (a) {};
\node[style=vertex] at (1, 0) (b) {};
\node[style=vertex] at (0.5, 0.71) (c) {};
\draw (a) -- (b);
\draw[style=cut-edge, mygreen] (b) -- (c);
\draw[style=cut-edge] (c) -- (a);
\node at (0.5, 0.15) {0};
\node at (0.15, 0.45) {1};
\node at (0.85, 0.45) {\color{mygreen}1};
\end{tikzpicture} \\
\begin{tikzpicture}
\node[style=vertex] at (0, 0) (a) {};
\node[style=vertex] at (1, 0) (b) {};
\node[style=vertex] at (0.5, 0.71) (c) {};
\draw[style=cut-edge] (a) -- (b);
\draw[style=cut-edge, mygreen] (b) -- (c);
\draw (c) -- (a);
\node at (0.5, 0.15) {1};
\node at (0.15, 0.45) {0};
\node at (0.85, 0.45) {\color{mygreen}1};
\end{tikzpicture}
& \begin{tikzpicture}
\node[style=vertex] at (0, 0) (a) {};
\node[style=vertex] at (1, 0) (b) {};
\node[style=vertex] at (0.5, 0.71) (c) {};
\draw[style=cut-edge] (a) -- (b);
\draw[style=cut-edge] (b) -- (c);
\draw[style=cut-edge] (c) -- (a);
\node at (0.5, 0.15) {1};
\node at (0.15, 0.45) {1};
\node at (0.85, 0.45) {\color{mygreen}1};
\end{tikzpicture}
\end{tabular}
\hfill
\begin{tikzpicture}[scale=0.9,line join=round]
\draw[arrows=-latex](0,0)--(-.731,-.422);
\draw[arrows=-latex](0,0)--(.914,-.338);
\draw[arrows=-latex](0,0)--(0,1.039);
\filldraw[fill=mygreen!50,fill opacity=0.5](0,0)--(.831,.637)--(-.665,.56)--cycle;
\filldraw[fill=mygreen!50,fill opacity=0.5](0,0)--(.831,.637)--(.166,.253)--cycle;
\filldraw[fill=mygreen!50,fill opacity=0.5](0,0)--(-.665,.56)--(.166,.253)--cycle;
\filldraw[fill=mygreen!50,fill opacity=0.5](.831,.637)--(-.665,.56)--(.166,.253)--cycle;
\node at (-.864,-.499) {$x_{e_1}$};  \node at (1.164,-.43) {$x_{e_2}$};  \node at (0,1.227) {$x_{\color{mygreen}f}$};
\end{tikzpicture}
\caption{Depicted on the left is a graph $G = (V, E)$ with $E = \{e_1, e_2\}$ and an augmentation $\widehat{G} = \allowbreak {(V, E \cup F)}$ of $G$ with $F = \{f\}$.
Depicted in the middle are the four feasible solutions to the lifted multicut problem with respect to $G$ and $\widehat{G}$.
Depicted on the right is the lifted multicut polytope $\Xi_{G \widehat{G}}$.
The figure is adopted from \citet{andres-2023-a-polyhedral}.}
\label{figure:polytope}
\end{figure}
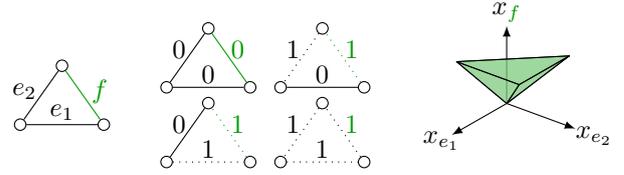

We analyze the convex hull $\Xi_{G \widehat{G}} := \conv X_{G\widehat{G}}$ of the feasible set $X_{G\widehat{G}}$ in the real affine space $\mathbb{R}^{E \cup F}$,
complementing properties established by \citet{hornakova-2017} and \citet{andres-2023-a-polyhedral} who call $\Xi_{G \widehat{G}}$ the \emph{lifted multicut polytope} with respect to $G$ and $\widehat{G}$
(An example of this polytope is depicted in \Cref{figure:polytope}).
More specifically, we establish a necessary, sufficient and efficiently decidable condition for an inequality $0 \leq x_e$ with $e \in E \cup F$ to define a facet of $\Xi_{G \widehat{G}}$.
Our proof involves an application of Menger's theorem \citep{menger-1927}.
In addition, we show: Deciding whether a cut inequality \eqref{eq:cut} defines a facet of $\Xi_{G \widehat{G}}$ is \textsc{np}-hard.

\section{Related Work}
\label{section:related-work}

The lifted multicut problem was introduced in the context of image and mesh segmentation by
\citet{keuper-2015a}
and is discussed in further detail by 
\citet{hornakova-2017} and \citet{andres-2023-a-polyhedral}.
It has diverse applications, notably to the tasks of image segmentation 
\citep{keuper-2015a, beier-2016, beier-2017},
video segmentation 
\citep{keuper-2017, keuper-2020},
and multiple object tracking 
\citep{tang-2017-multiple, nguyen-2022, kostyukhin-2023}.
For these applications, local search algorithms are defined, implemented and compared empirically by \citet{keuper-2015a,levinkov-2017}.
Two branch-and-cut algorithms for the lifted multicut problem are defined, implemented and compared empirically by \citet{hornakova-2017}.

In order to significantly reduce the runtime of their branch-and-cut algorithm, 
\citet{hornakova-2017} are also the first to establish properties of lifted multicut polytopes, including its dimension $\dim \Xi_{G \widehat{G}} = |E \cup F|$ and a characterization of facets induced by cycle inequalities \eqref{eq:cycle}, path inequalities \eqref{eq:path}, upper box inequalities $x_e \leq 1$ for $e \in E \cup F$, and lower box inequalities $0 \leq x_e$ for $e \in E$.
Moreover, they establish necessary conditions on facets of lifted multicut polytopes induced by cut inequalities \eqref{eq:cut} and lower box inequalities $0 \leq x_e$ for $e \in F$.
\citet{andres-2023-a-polyhedral} describe an additional class of facets induced by so-called half-chorded odd cycle inequalities and show that these are facets also of a polytope isomorphic to the clique partitioning polytope \cite{groetschel-1990,deza-1992,deza-1997,sorensen-2002}.
Additionally, they establish the class of facets induced by so-called intersection inequalities, which is discovered based on a necessary condition for facets induced by cut inequalities \eqref{eq:cut}.
However, they do not make progress toward characterizing the facets of lifted multicut polytopes induced by cut inequalities themselves  
or lower box inequalities $0 \leq x_e$ for $e \in F$, which motivates the work we show in this article.

\section{Preliminaries}
\label{section:preliminaries}

For clarity, we adopt elementary terms and notation:
Let be $G = (V,E)$ a graph. 
For any distinct $u, w \in V$, we write $uw$ and $wu$ as an abbreviation of the set $\set{u, w}$.
For any subset $A \subseteq E$, we write $\mathbbm{1}_A \in \set{0,1}^E$ for the characteristic vector of the set $A$, i.e. $(\mathbbm{1}_A)_e = 1 \Leftrightarrow e \in A$ for all $e \in E$.
For any distinct $u, w \in V$ and any $S \subseteq V$, we call $S$ a \emph{$uw$-separator} of $G$ and say that $u$ and $w$ are \emph{separated} by $S$ in $G$ if and only if every $uw$-path in $G$ contains a node of $S$.
We call $S$ \emph{proper} if and only if $u \notin S$ and $w \notin S$.
For any $u,v,w \in V$ such that $S = \{v\}$ is a $uw$-separator of $G$, we call $v$ a \emph{$uw$-cut-node} of $G$.
We call it \emph{proper} if and only if $S$ is \emph{proper}.
We let $C_{uw}(G)$ denote the set of all proper $uw$-cut-nodes of $G$.

Next, we discuss briefly the inequalities \eqref{eq:cycle}--\eqref{eq:cut} from \Cref{definition:lifted-multicut-problem} and refer to \citet[Proposition~3]{andres-2023-a-polyhedral} for details:
The inequalities \eqref{eq:cycle} state that no cycle in $G$ intersects with the set $\{e \in E \mid x_e = 1\}$ in precisely one edge.
This property is equivalent to the existence of a decomposition $\Pi$ of $G$ such that for any $uw \in E$: $x_{uw} = 0$ if and only if there exists a component $U \in \Pi$ such that $uw \subseteq U$.
The inequalities \eqref{eq:path} and \eqref{eq:cut} together state for any $uw \in F$ that $x_{uw} = 0$ if and only if there exists a $uw$-path $(V_P, E_P)$ in $G$ with all edges $e \in E_P$ such that $x_e = 0$, i.e.~that $x_{uw} = 0$ if and only if there exists a component $U \in \Pi$ such that $uw \subseteq U$.

Below, we establish one consequence of these properties that we apply in this article:
\begin{lemma}
    \label{lemma:construction}
    For 
    any connected graph $G = (V, E)$,
    any augmentation $\widehat{G} = (V, E \cup F)$ with $F \cap E = \emptyset$,
    any disjoint node sets $A \subseteq 2^V$ 
    such that for all $U \in A$ the subgraph $G[U]$ of $G$ induced by $U$ is connected and
    any $x^A \in \{0,1\}^{E \cup F}$ such that for all $uw \in E \cup F$, $x^A_{uw} = 0 \Leftrightarrow \exists U \in A \colon uw \subseteq U$, we have
    $x^A \in X_{G\widehat{G}}$.
\end{lemma}
\begin{proof}
    Firstly, $\Pi = (A \setminus \set{\emptyset}) \cup \set{\set{v} \mid v \in V \setminus \cup_{U \in A} U}$ is a decomposition of $G$.
    Secondly, $x^A$ is such that for any $uw \in E \cup F$ we have $x^A_{uw} = 0$ if and only if there is a $U \in \Pi$ such that $uw \subseteq U$.
    Thus, $x^A \in X_{G \widehat{G}}$.
\end{proof}

\section{Lower Box Facets}
\label{section:box-inequalities}

In this section, we establish a necessary, sufficient and efficiently decidable condition for a lower box inequality $0 \leq x_{uw}$ with $uw \in E \cup F$ to define a facet of a lifted multicut polytope $\Xi_{G \widehat{G}}$.
Examples are shown in \Cref{figure:box-necessity}.

\begin{theorem}\label{theorem:box-inequalities}
For any connected graph $G = (V,E)$, 
any augmentation $\widehat{G} = (V, E \cup F)$ with $F \cap E = \emptyset$
and any $uw \in E \cup F$,
the lower box inequality $0 \leq x_{uw}$ is facet-defining for $\Xi_{G \widehat{G}}$ if and only if the following two conditions hold:
\begin{enumerate}
\item There exists no simple path in $\widehat{G}$ of length at least one, besides $(\set{u,w}, \{uw\})$, whose end-nodes are $uw$-cut-nodes of $G$ and whose edges are $uw$-separators of $G$.
\label{theorem:condition:path}
\item There exists no simple cycle in $\widehat{G}$ whose edges are $uw$-separators of $G$.
\label{theorem:condition:cycle}
\end{enumerate}
\end{theorem}

\begin{figure}[t]
	\centering
	\begin{subfigure}{0.4\textwidth}
		\centering
		\resizebox{!}{!}{
			\begin{tikzpicture}
    \node[vertex, label=left:$u$] (u) at (0, 0) {};
    \node[vertex, label=below:$v_0$] (v0) at (1, -0.75) {};
    \node[vertex, label=above:$v_1$] (v1) at (1, 0.75) {};
    \node[vertex, label=below:$v_2$] (v2) at (2.5, -0.75) {};
    \node[vertex, label=above:$v_3$] (v3) at (2.5, 0.75) {};
    \node[vertex, label=below:$v_4$] (v4) at (4, -0.75) {};
    \node[vertex, label=above:$v_5$] (v5) at (4, 0.75) {};
    \node[vertex, label=right:$w$] (w) at (5, 0) {};

    \node[vertex] (u1) at (0.5, 0.375) {};
    \node[vertex] (u3) at (4.5, 0.375) {};

    \draw[mygreen][thick] (u) -- (v0);
    \draw (u) -- (u1) --(v1);
    \draw (v1) -- (v3);
    \draw (v0) -- (v2);
    \draw[myred][thick] (v0) -- (v1);
    \draw[dashed][myred][thick] (v2) -- (v3);
    \draw (v4) -- (v5);
    \draw[dashed][mygreen][thick] (v1) -- (v2);
    \draw[dashed][mygreen][thick] (v3) -- (v4);
    \draw (v2) -- (v4);
    \draw (v3) -- (v5);
    \draw[myred][thick] (v4) -- (w);
    \draw (v5) -- (u3) -- (w);

    \draw[dashed] (u) .. controls +(0,2) and +(0,2) .. (w) node {};
\end{tikzpicture}
		}
		\caption{$0 = x_{uv_0} - x_{v_0v_1} + x_{v_1v_2} - x_{v_2v_3} + x_{v_3v_4} - x_{v_4w}$}
	\end{subfigure}
	\begin{subfigure}{0.4\textwidth}
		\centering
		\resizebox{!}{!}{
			\begin{tikzpicture}
    \node[vertex, label=left:$u$] (u) at (0, 0) {};
    \node[vertex, label=below:$v_0$] (v0) at (1, -0.75) {};
    \node[vertex, label=above:$v_1$] (v1) at (1, 0.75) {};
    \node[vertex, label=below:$v_2$] (v2) at (2.5, -0.75) {};
    \node[vertex, label=above:$v_3$] (v3) at (2.5, 0.75) {};
    \node[vertex, label=below:$v_4$] (v4) at (4, -0.75) {};
    \node[vertex, label=above:$v_5$] (v5) at (4, 0.75) {};
    \node[vertex, label=right:$w$] (w) at (5, 0) {};

    \node[vertex] (u0) at (0.5, -0.375) {};
    \node[vertex] (u1) at (0.5, 0.375) {};
    \node[vertex] (u2) at (4.5, -0.375) {};
    \node[vertex] (u3) at (4.5, 0.375) {};

    \draw (u) -- (u0) -- (v0);
    \draw (u) -- (u1) -- (v1);
    \draw (v1) -- (v3);
    \draw (v0) -- (v2);
    \draw[mygreen][thick] (v0) -- (v1);
    \draw[dashed][mygreen][thick] (v2) -- (v3);
    \draw[mygreen][thick] (v4) -- (v5);
    \draw[dashed][myred][thick] (v1) -- (v2);
    \draw[dashed][myred][thick] (v3) -- (v4);
    \draw[dashed][myred][thick] (v5) -- (v0);
    \draw (v2) -- (v4);
    \draw (v3) -- (v5);
    \draw (v4) -- (u2) -- (w);
    \draw (v5) -- (u3) -- (w);

    \draw[dashed] (u) .. controls +(0,2) and +(0,2) .. (w) node {};
\end{tikzpicture}
		}
		\caption{$0 = x_{v_0v_1} - x_{v_1v_2} + x_{v_2v_3} - x_{v_3v_4} + x_{v_4v_5} - x_{v_5v_0}$}
	\end{subfigure}
	\caption{
		Depicted above are two examples of a graph $G$ (solid edges) and augmentation $\widehat{G}$ (dashed edges) such that a condition of \Cref{theorem:box-inequalities} is violated for the inequality $0 \leq x_{uw}$.
		In the upper example, the path with the edge set $\{uv_0, v_0v_1, v_1v_2, v_2v_3, v_3v_4, v_4w\}$ violates \ref{theorem:condition:path}. 
		In the lower example, the cycle with the edge set $\{v_0v_1, v_1v_2, v_2v_3, v_3v_4, v_4v_5, v_5v_0\}$ violates \ref{theorem:condition:cycle}. 
		For both cases, Equation \eqref{equation:equality} from the proof of \Cref{theorem:box-inequalities} is stated.
		Edges depicted in green occur with a positive sign in this equation, and edges depicted in red occur with a negative sign.
    }
	\label{figure:box-necessity}
\end{figure}

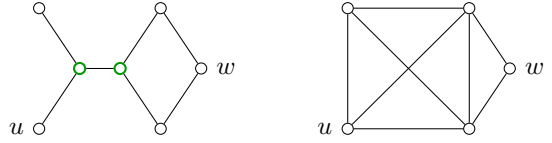
\begin{figure}[t]
	\centering
	\begin{subfigure}{.23\textwidth}
		\centering
		\resizebox{.90\textwidth}{!}{
			\begin{tikzpicture}
    \node[vertex, label=left:$u$] (u) at (0, -0.75) {};
    \node[vertex] (v0) at (0, 0.75) {};
    \node[vertex, color=mygreen, fill=white, thick] (c1) at (0.5, 0) {};
    \node[vertex, color=mygreen, fill=white, thick] (c2) at (1, 0) {};
    \node[vertex] (v2) at (1.5, 0.75) {};
    \node[vertex] (v3) at (1.5, -0.75) {};
    \node[vertex, label=right:$w$] (w) at (2, 0) {};

    \draw (u) -- (c1) -- (c2);
    \draw (c1) -- (v0);
    \draw (c2) -- (v2) -- (w);
    \draw (c2) -- (v3) -- (w);
\end{tikzpicture}
		}
	\end{subfigure}
	\begin{subfigure}{.23\textwidth}
		\centering
		\resizebox{.90\textwidth}{!}{
			\begin{tikzpicture}
    \node[vertex, label=left:$u$] (u) at (0, -0.75) {};
    \node[vertex] (v0) at (0, 0.75) {};
    \node[vertex] (v1) at (1.5, -0.75) {};
    \node[vertex] (v2) at (1.5, 0.75) {};
    \node[vertex, label=right:$w$] (w) at (2, 0) {};

    \draw (u) -- (v0);
    \draw (u) -- (v1);
    \draw (u) -- (v2);

    \draw (v0) -- (v1);
    \draw (v0) -- (v2);

    \draw (v1) -- (v2);
    \draw (v1) -- (w);
    \draw (v2) -- (w);
\end{tikzpicture}
		}
	\end{subfigure}
	\caption{
		Depicted on the left is a graph $G$, and depicted on the right is the corresponding auxiliary graph $G'$, whose construction is described in the proof of \Cref{lemma:cycle-even-length}. 
		The nodes depicted in green are proper $uw$-cut-nodes of $G$ and get removed in the construction of $G'$.
	}
	\label{figure:even-cycle}
\end{figure}

In the remainder of this section, we prove a structural lemma and then apply this lemma in order to prove \Cref{theorem:box-inequalities}.

\begin{lemma}\label{lemma:cycle-even-length}
Let $G = (V, E)$ be a graph and let $u,w \in V$.
Any simple cycle $(V_C, E_C)$ with $V_C \subseteq V$ and $E_C \subseteq \tbinom{V}{2}$ such that no $v \in V_C$ is a $uw$-cut-node of $G$ and every $e \in E_C$ is a $uw$-separator of $G$ is even.
\end{lemma}
\begin{proof}[Proof of \Cref{lemma:cycle-even-length}]
In a first step, we construct for any $u, w \in V$ and any cycle $C = (V_C,E_C)$ as defined in the lemma an auxiliary graph $G' = (V', E')$ by removing from $G$ the set $C_{uw}(G)$ of all proper $uw$-cut-nodes and connecting remaining nodes, which are connected in $G$ by a path of only proper $uw$-cut-nodes, by additional edges, i.e.
\begin{align*}
V' =\ & V \setminus C_{uw}(G)
\\
E' =\  
& \big\{
st \in \tbinom{V'}{2} \,\big|\, \exists (V_P,E_P) \in st\text{-paths}(G) \colon \\ &
\hspace{1ex} V_P \setminus st \subseteq C_{uw}(G) \big\} \cup \big(E \cap \tbinom{V'}{2}\big) 
\enspace .
\end{align*}
An example of this construction is shown in \Cref{figure:even-cycle}.

In a second step, we now show that $G'$ has the following properties:
\begin{enumerate}[label=\textnormal{(\roman*)}]
\item $V_C \cup \set{u,w} \subseteq V'$ and $E_C \subseteq \tbinom{V'}{2}$;
\label{item:G-prime-1}
\item there exist no proper $uw$-cut-nodes of $G'$;
\label{item:G-prime-2}
\item all $e \in E_C$ are proper $uw$-separators of $G'$.
\label{item:G-prime-3}
\end{enumerate}

Property \ref{item:G-prime-1} follows directly from the construction of the auxiliary graph $G'$.

Assume \ref{item:G-prime-2} does not hold.
Then there exists a $v \in C_{uw}(G')$.
It follows $v \not\in \set{u,w}$ and, by construction of $G'$, that $v \not\in C_{uw}(G)$.
Thus, $v$ is no $uw$-cut-node of $G$ and there exists a $uw$-path $(V_P, E_P)$ in $G$ such that $v \not\in V_P$.
By construction of $G'$, we can create a $uw$-path $(V_{P'}, E_{P'})$ in $G'$ with $v \not\in V_{P'}$ by replacing all subpaths of $(V_{P}, E_{P})$ whose internal nodes are in $C_{uw}(G)$ with edges in $E' \setminus E$.
The existence of such a $uw$-path $(V_{P'}, E_{P'})$ contradicts $v \in C_{uw}(G')$.

Assume \ref{item:G-prime-3} does not hold. 
Then there exists an $e \in E_C$ that is not a proper $uw$-separator of $G'$.
As $e \cap \set{u,w} = \emptyset$ by assumption, $e$ is also no $uw$-separator of $G'$.
Thus, there exists a $uw$-path $(V_{P'}, E_{P'})$ in $G'$ with $e \cap V_{P'} = \emptyset$.
By construction of $G'$, we can create a $uw$-path $(V_{P}, E_{P})$ in $G$ from $(V_{P'}, E_{P'})$ by replacing all edges in $E_{P'} \setminus E$ with paths in $G$ whose internal nodes are in $C_{uw}(G)$.
For this path $P$, it holds $e \cap V_P = \emptyset$ because $e \cap V_{P'} = \emptyset$ (see above) and $e \cap C_{uw}(G) = \emptyset$ (by assumption).
The existence of such a $uw$-path $(V_{P}, E_{P})$ contradicts $e$ being a $uw$-separator of $G$.

In a third step, we now prove that $C$ is even:
Menger's theorem \citep{menger-1927} states that for two distinct non-adjacent nodes $a,b \in V'$, the number of internally node-disjoint $ab$-paths in $G'$ is equal to the minimal size of proper $ab$-separators of $G'$.
By \ref{item:G-prime-1}, $u$ and $w$ are in $V'$.
Furthermore, they are distinct and non-adjacent in $G'$, as otherwise every $uw$-separator of $G'$ would contain $u$ or $w$, in contradiction to the elements of $E_C$ being proper $uw$-separators of $G'$ by \ref{item:G-prime-3}.
As $C_{uw}(G') = \emptyset$ by \ref{item:G-prime-2}, and all edges in $E_C$ are proper $uw$-separators of $G'$ by \ref{item:G-prime-3}, the minimal size of proper $uw$-separators of $G'$ is two.
Thus, there exist precisely two internally node-disjoint $uw$-paths $P_1 = (V_{P_1}, E_{P_1})$ and $P_2 = (V_{P_2}, E_{P_2})$ in $G'$, by Menger's theorem.

W.l.o.g., we enumerate the nodes in the cycle $(V_C, E_C)$:
For $n := |V_C|$, let $v \colon \mathbb{Z}_n \to V_C$ such that $E_C = \{ v_j v_{j+1} \mid j \in \mathbb{Z}_n \}$.
As $v_0 v_1$ is a $uw$-separator of $G'$ by \ref{item:G-prime-3}, the paths $P_1$ and $P_2$ each contain $v_0$ or $v_1$. 
Moreover, as these paths are internally node-disjoint, precisely one of them contains $v_0$, the other $v_1$.
Assume w.l.o.g.~that $v_0 \in V_{P_1}$ and $v_1 \in V_{P_2}$.
By \ref{item:G-prime-3}, any $v_j v_{j+1} \in E_C$ with $j \in \set{1,\ldots,n-2}$ is a $uw$-separator of $G'$.
Thus:
\begin{align*}
V_{P_1} \cap V_C &= \bigl\{v_{2j} \mid j \in \{0, \ldots, \lfloor \tfrac{n-1}{2} \rfloor\}\bigr\}\\
V_{P_2} \cap V_C &= \bigl\{v_{2j+1} \mid j \in \{0, \ldots, \lfloor \tfrac{n-2}{2} \rfloor\}\bigr\} \enspace .
\end{align*}
If $C$ were odd, $n$ would be odd.
Thus, $n-1$ would be even.
Consequently, it would follow that ${v_{n-1}v_0 \cap V_{P_2} = \emptyset}$, in contradiction to $v_{n-1}v_0$ being a $uw$-separator of $G'$ by \ref{item:G-prime-3}.
Thus, $C$ must be even.
\end{proof}

\begin{proof}[Proof of \Cref{theorem:box-inequalities}]
Assume there exists a path or cycle $H = (V_H, E_H)$ of $uw$-separators of $G$ as defined in the theorem. 

W.l.o.g., fix enumerations of the nodes and edges of $H$ as follows:
Let $n := |E_H|$.
If $H$ is a path, let $v \colon \{0, \ldots, n\} \to V_H$ and $e \colon \{0, \ldots, n-1\} \to E_H$ such that $\forall j \in \{0, \ldots, n-1\} \colon e_j = v_j v_{j+1}$ and $E_H = {\{ e_j \mid j \in \{0, \ldots, n-1\}\}}$.
If $H$ is a cycle, let $v \colon \mathbb{Z}_n \to V_H$ and $e \colon \mathbb{Z}_n \to E_H$
such that $\forall j \in \mathbb{Z}_n \colon e_j = v_j v_{j+1}$ and $E_H = \{ e_j \mid j \in \mathbb{Z}_n \}$.
If $H$ is a cycle containing $uw$-cut-nodes of $G$, assume further and w.l.o.g.~that $v_0 = v_n$ is such a $uw$-cut-node.
Finally, consider the partition $\set{E_0, E_1}$ of $E_H$ into even and odd edges, i.e.
\begin{align*}
E_0 & = \left\{ e_{2j} \in E_H \;\middle|\; j \in \{0, \ldots, \lfloor \tfrac{n-1}{2} \rfloor \right\} 
\\
E_1 & = \left\{ e_{2j+1} \in E_H \;\middle|\; j \in \{0, \ldots, \lfloor \tfrac{n-2}{2} \rfloor \right\}
\enspace .
\end{align*}
We will prove that $\sigma=\set{x \in X_{G \widehat{G}} \mid x_{uw} = 0}$ is not a facet of $\Xi_{G \widehat{G}}$ by showing that all $x \in \sigma$ satisfy the additional orthogonal equality
\begin{align}
0 = \sum_{j \in \set{0, \ldots, n-1}} (-1)^j \, x_{e_j}
\enspace .
\label{equation:equality}
\end{align}
More specifically, we will prove for every $x \in \sigma$ the existence of a bijection 
\begin{equation*}
\vartheta_x \colon\; E_0 \cap x^{-1}(1) \;\rightarrow\; E_1 \cap x^{-1}(1)
\enspace .
\end{equation*}
Using these bijections, we conclude for every $x \in \sigma$ that the number of elements in the sum of (\ref{equation:equality}) taking the value $+1$ is equal to the number of elements taking the value $-1$, and thus that the equality holds.

We now show that these bijections exist.
Let $x \in \sigma$. 
As $x_{uw}=0$, the decomposition of $G$ induced by $x$ has a component containing both $u$ and $w$. 
Let $V_{uw}$ be the node set of that component.

If $n = 1$, then $H$ is a path $(\set{v_0,v_1}, \set{e_0})$. Thus, $E_1 \cap x^{-1}(1) = \emptyset$ because $E_1 = \emptyset$.
Moreover, $E_0 \cap x^{-1}(1) = \emptyset$ as $v_0$ and $v_1$ are $uw$-cut-nodes of $G$ and thus elements of $V_{uw}$, which implies $x_{e_0} = 0$.
In this case, $\vartheta_x = \emptyset$ and \eqref{equation:equality} specializes to $x_{e_0} = 0$, which is satisfied.

We now consider $n \geq 2$.
For every $e_j = v_j v_{j+1} \in E_0 \cap x^{-1}(1)$, we define:
\begin{equation}
	\vartheta_x(e_j) =
	\begin{cases}
		e_{j-1} & \text{if}\ v_j \not\in V_{uw} \\
		e_{j+1} & \text{if}\ v_{j+1} \not\in V_{uw}
	\end{cases}
	\enspace .
\end{equation}

We show that $\vartheta_x$ is well-defined:
Let $e_j \in E_0 \cap x^{-1}(1)$.
In general, at least one of $v_j$ and $v_{j+1}$ is not in $V_{uw}$ because $x_{e_j} = 1$, and at most one of $v_j$ and $v_{j+1}$ is not in $V_{uw}$ because $e_j$ is a $uw$-separator of $G$.
Thus, $\vartheta_x$ assigns $e_j$ a unique element, it remains to show that this element is in $E_1 \cap x^{-1}(1)$.

Firstly, we show $\vartheta_x(e_j) \in E_1$.
Clearly, it holds for $j \in \set{1, \ldots, n-2}$, that $\vartheta_x(e_j) \in E_1$.
We regard the remaining cases of $j \in \{0, n-1\}$. 
Let first $j = 0$.
For $H$ a path or cycle with $uw$-cut-node, $\vartheta_x(e_0) = e_1 \in E_1$ because $v_0 \in V_{uw}$ as $v_0$ is a $uw$-cut-node of $G$.
For $H$ a cycle without $uw$-cut-node, we distinguish $v_0 \in V_{uw}$ and $v_0 \not\in V_{uw}$.
If $v_0 \in V_{uw}$, then $\vartheta_x(e_0) = e_1 \in E_1$.
If $v_0 \not\in V_{uw}$, then $\vartheta_x(e_0) = e_{n-1} \in E_1$ because $n-1$ is odd by \Cref{lemma:cycle-even-length}.
Let now $j = n-1$.
For $H$ a path or cycle with $uw$-cut-node, $\vartheta_x(e_{n-1}) = e_{n-2} \in E_1$ because $v_{n} \in V_{uw}$ as $v_{n}$ is a $uw$-cut-node of $G$.
For $H$ a cycle without $uw$-cut-node, $e_{n-1} \not\in E_0 \cap x^{-1}(1)$ because $n-1$ is odd by \Cref{lemma:cycle-even-length}.

Secondly, we show $x_{\vartheta_x(e_j)} = 1$.
By definition of $\vartheta_x$, $e_j$ and $\vartheta_x(e_j)$ share a node $v \notin V_{uw}$.
As $\vartheta_x(e_j)$ is a $uw$-separator of $G$, the other node of $\vartheta_x(e_j)$ is in $V_{uw}$ and therefore $x_{\vartheta_x(e_j)} = 1$.
Thus, it holds $\vartheta_x(e_j) \in E_1 \cap x^{-1}(1)$, and $\vartheta_x$ is well-defined.

We show that $\vartheta_x$ is surjective:
Let $e_j \in E_1 \cap x^{-1}(1)$.
As $x_{e_j}=1$, either $v_j \not\in V_{uw}$ or $v_{j+1} \not\in V_{uw}$.
If $v_j \not\in V_{uw}$, then $e_{j-1} \in E_0 \cap x^{-1}(1)$ and $\vartheta_x(e_{j-1}) = e_j$.
If $v_{j+1} \not\in V_{uw}$, then $e_{j+1} \in E_0 \cap x^{-1}(1)$ and $\vartheta_x(e_{j+1}) = e_j$.
Thus, $\vartheta_x$ is surjective.

We show that $\vartheta_x$ is injective:
Assume $\vartheta_x$ is not injective.
Then there exists a $j \in \set{0, \ldots, n-1}$ such that $e_j \in E_1 \cap x^{-1}(1)$ and $e_{j-1}, e_{j+1} \in E_0 \cap x^{-1}(1)$ such that $\vartheta_x(e_{j-1}) = e_j = \vartheta_x(e_{j+1})$, by definition of $\vartheta_x$.
This implies $v_j,v_{j+1} \not\in V_{uw}$, which contradicts $e_j$ being a $uw$-separator.
By this contradiction, $\vartheta_x$ is injective.

Altogether, we have shown that $\vartheta_x$ is well-defined, surjective and injective, and thus a bijection.
This concludes the proof of necessity.

Assume now that \ref{theorem:condition:path} and \ref{theorem:condition:cycle} are satisfied.
We prove that ${0 \leq x_{uw}}$ is facet-defining by constructing ${\lvert E \cup F \rvert - 1}$ linearly independent vectors in the linear space $\lin(\set{x-y \mid x,y \in \sigma})$ which we abbreviate by $\lin(\sigma-\sigma)$, implying $\dim \aff \sigma = \dim \lin(\sigma - \sigma) = \lvert E \cup F \rvert - 1$. In particular, we construct the characteristic vectors of all $st \in E \cup F \setminus \set{uw}$.
For this construction, we distinguish the following cases:
\begin{enumerate}
\item $st$ is not a $uw$-separator of $G$;
\item $st$ is a $uw$-separator of $G$ and neither $s$ nor $t$ is a $uw$-cut-node;
\item precisely one node of $st$ is a $uw$-cut-node.
\end{enumerate}
Note that no $st \in E \cup F \setminus \set{uw}$ is such that both $s$ and $t$ are $uw$-cut-nodes, as otherwise the path $(\set{s,t}, \set{st})$ would violate \ref{theorem:condition:path}.
Thus, this distinction of cases is complete.

For the first case, let $st \in E \cup F \setminus \set{uw}$ such that $st$ is not a $uw$-separator of $G$.
By this property, there exists a $uw$-path $(V_{P_{uw}}, E_{P_{uw}})$ in $G$ that contains neither $s$ nor $t$.
Let further $(V_{P_{st}}, E_{P_{st}})$ be an $st$-path in $G$.
If $G[V_{P_{uw}} \cup V_{P_{st}}]$ is not connected, we define:
\begin{gather}
\begin{align*}
	V_1 &= \{V_{P_{uw}}, V_{P_{st}}\} &
	V_2 &= \{V_{P_{uw}}, V_{P_{st}} \setminus \set{s, t}\} \\ 
	V_3 &= \{V_{P_{uw}}, V_{P_{st}} \setminus \set{s}\} &
	V_4 &= \{V_{P_{uw}}, V_{P_{st}} \setminus \set{t}\}
	\enspace .
\end{align*}
\end{gather}
Otherwise, we define:
\begin{gather}
\begin{align*}
	V_1 &= \{V_{P_{uw}} \cup V_{P_{st}}\} &
	V_2 &= \{V_{P_{uw}} \cup V_{P_{st}} \setminus \set{s, t}\} \\ 
	V_3 &= \{V_{P_{uw}} \cup V_{P_{st}} \setminus \set{s}\} &
	V_4 &= \{V_{P_{uw}} \cup V_{P_{st}} \setminus \set{t}\}
	\enspace .
\end{align*}
\end{gather}
In both cases, it is easy to see for $i \in \set{1, \ldots, 4}$ that $G[U]$ is connected for all $U \in V_i$ and thus $x^{V_i} \in X_{G\widehat{G}}$, by \Cref{lemma:construction}.
It further holds, $\mathbbm{1}_\set{st} = -x^{V_1}-x^{V_2}+x^{V_3}+x^{V_4}$ and $x^{V_i}_{uw} = 0$, as for all $pq \in E \cup F$:
\begin{itemize}
	\item $x^{V_i}_{pq} = 1 \ \text{for} \ i = 1,\ldots, 4 \ 
	\text{if} \ \nexists U \in V_1\colon \set{p,q} \subseteq U$
	\item $x^{V_i}_{pq} = 0 \ \text{for} \ i = 1,\ldots, 4 \ 
	\text{if} \ \exists U \in V_2\colon \set{p,q} \subseteq U$
	\item $x^{V_i}_{pq} = 0 \ \text{for} \ i = 1,3 \ \text{and} \ x^{V_i}_{pq} = 1 \ \text{for} \ i = 2,4 \\ 
	\text{if} \ s \in \set{p,q}\,, t \not\in \set{p,q} \;\text{and}\; \exists U \in V_1\colon \set{p,q} \subseteq U$
	\item $x^{V_i}_{pq} = 0 \ \text{for} \ i = 1,4 \ \text{and} \ x^{V_i}_{pq} = 1 \ \text{for} \ i = 2,3 \\ 
	\text{if} \ t \in \set{p,q}\,, s \not\in \set{p,q} \;\text{and}\; \exists U \in V_1\colon \set{p,q} \subseteq U$
	\item $x^{V_1}_{pq} = 0 \ \text{and} \ x^{V_i}_{pq} = 1 \ \text{for} \ i = 2,3,4 \ \text{if} \ \set{p,q} = \set{s,t}$\,.
\end{itemize}
It follows from $x^{V_i}_{uw} = 0$ that  $x^{V_i} \in \sigma$. 
Thus, $\mathbbm{1}_\set{st} = -x^{V_1}-x^{V_2}+x^{V_3}+x^{V_4} \in \lin(\sigma - \sigma)$, which concludes the first case.

For the second case, consider the set $H$ of all $st \in E \cup F \setminus \set{uw}$ such that $st$ is a $uw$-separator of $G$ and neither $s$ nor $t$ is a $uw$-cut-node of $G$.
Let $st \in H$ and let $v \in st$.
As $v$ is no $uw$-cut-node, there exists a $uw$-path $(V_{P_{uw}}, E_{P_{uw}})$ in $G$ that does not contain $v$.
Let further $(V_{P_{st}}, E_{P_{st}})$ be an $st$-path in $G$ and let $P = (V_P, E_P) = (V_{P_{uw}} \cup V_{P_{st}}, E_{P_{uw}} \cup E_{P_{st}})$.
With $E_{\widehat{G}}(P,v) = \set{vv' \in (E \cup F) \cap \tbinom{V_P}{2} }$ denoting the set of edges of $\widehat{G}$ containing $v$ whose nodes are in $V_P$, we first show that $\mathbbm{1}_{E_{\widehat{G}}(P,v)} \in \lin(\sigma - \sigma)$.
If $G[V_{P_{uw}} \cup V_{P_{st}}]$ is not connected, we define:
\begin{equation*}
	V_1 = \{V_{P_{uw}}, V_{P_{st}}\} \quad
	V_4 = \{V_{P_{uw}}, V_{P_{st}} \setminus \set{t}\} 
	\enspace .
\end{equation*}
Otherwise, we define: 
\begin{equation*}
	V_1 = \{V_{P_{uw}} \cup V_{P_{st}}\} \quad 
	V_4 = \{V_{P_{uw}} \cup V_{P_{st}} \setminus \set{t}\} 
	\enspace .
\end{equation*}
Analogously to the previous case, we get $\mathbbm{1}_{E_{\widehat{G}}(P,v)} = {-x^{V_1}+x^{V_4}} \in \lin(\sigma - \sigma)$.
Denoting by $E_{\widehat{G}}(v) = \set{vv' \in (E \cup F)}$ the set of edges of $\widehat{G}$ containing $v$ and noting that
\begin{equation}\label{equation:adjacent-edges}
		\mathbbm{1}_\set{st} = \mathbbm{1}_{E_{\widehat{G}}(P,v)} - \sum_{e \in E_{\widehat{G}}(P,v)\setminus \set{st}} \mathbbm{1}_\set{e}
\end{equation}
and $E_{\widehat{G}}(P,v) \subseteq E_{\widehat{G}}(v)$, we see that it is sufficient for proving $\mathbbm{1}_\set{st} \in \lin(\sigma - \sigma)$ to show that there exists a node $v \in st$ such that $\mathbbm{1}_\set{e} \in \lin(\sigma - \sigma)$ for all $e \in E_{\widehat{G}}(v) \setminus \set{st}$.

Next, we define a sequence $\set{H_j}_{j \in \mathbb{N}_0}$ of subsets of $H$ (for an example see \Cref{figure:Hj}) and show iteratively that the characteristic vectors of their elements are in $\lin(\sigma - \sigma)$ using \eqref{equation:adjacent-edges}.
For any $j \in \mathbb{N}_0$, we define:
\begin{align}\label{equation:definition-Hi}
	\begin{split}
		H_j = \bigl\{
			& st \in H \setminus \cup_{k<j} H_k \mid \exists v \in st\,\forall e \in E_{\widehat{G}}(v) \setminus \set{st} \colon\\ 
			& e \ \text{is no $uw$-separator of $G$} \vee\
			e \in \cup_{k<j} H_k
			\bigr\} \enspace .
	\end{split}
\end{align}
By this definition, for any $st \in H_0$ there exists a $v \in st$ such that all $e \in E_{\widehat{G}}(v) \setminus \set{st}$ are no $uw$-separators of $G$. Thus, it follows from the previous case that $\mathbbm{1}_\set{e} \in \lin(\sigma - \sigma)$ for all $e \in E_{\widehat{G}}(v) \setminus \set{st}$. Consequently, $\mathbbm{1}_\set{st} \in \lin(\sigma - \sigma)$ by \eqref{equation:adjacent-edges}.
Let now $j>0$ and assume that the characteristic vectors of all elements in $\cup_{k<j} H_k$ are in $\lin(\sigma - \sigma)$.
By definition, for any $st \in H_j$ there exists a $v \in st$ such that any $e \in E_{\widehat{G}}(v) \setminus \set{st}$ is either no $uw$-separator of $G$ and thus $\mathbbm{1}_\set{e} \in \lin(\sigma - \sigma)$ by the previous case, or is in $\cup_{k<j} H_k$ and thus ${1}_\set{e} \in \lin(\sigma - \sigma)$ by assumption.
Consequently, $\mathbbm{1}_\set{st} \in \lin(\sigma - \sigma)$ by \eqref{equation:adjacent-edges}.

For completing the second case, it remains to show that we have constructed the characteristic vectors of all elements in $H$ by this, i.e. that $H \subseteq \cup_{j \geq 0} H_j$.
This follows directly from \Cref{claim:H-partition}, which is proven in the appendix.
\begin{claim}\label{claim:H-partition}
	If \ref{theorem:condition:path} and \ref{theorem:condition:cycle} are satisfied, the set $\set{H_j \mid j \in \mathbb{N}_0 \wedge H_j \neq \emptyset}$ is a partition of $H$.
\end{claim}
This concludes the second case.

For the last case, let $st \in E \cup F \setminus \set{uw}$ such that precisely one node of $st$, say $t$, is a $uw$-cut-node.
We show $\mathbbm{1}_{E_{\widehat{G}}(s)} \in \lin(\sigma - \sigma)$ analogously to the previous case and again have
\begin{align}\label{eq:adjacent-edges-s}
\mathbbm{1}_\set{st} = \mathbbm{1}_{E_{\widehat{G}}(s)} - \sum_{e \in E_{\widehat{G}}(s)\setminus \set{st}} \mathbbm{1}_\set{e} 
\enspace .
\end{align}
For any $e = s's \in E_{\widehat{G}}(s)\setminus \set{st}$, $s'$ is no $uw$-cut-node of $G$, as otherwise the path $(\set{s',s,t},\set{s's, st})$ would violate \ref{theorem:condition:path}.
Consequently, $\mathbbm{1}_\set{e} \in \lin(\sigma - \sigma)$ by the previous two cases.
It follows from \eqref{eq:adjacent-edges-s} that $\mathbbm{1}_\set{st} \in \lin(\sigma - \sigma)$, which concludes the third case.
Altogether, we have constructed $\lvert E \cup F \rvert -1$ linearly independent vectors in $\lin(\sigma -\sigma)$ and have thus established sufficiency of the specified conditions.
\end{proof}

\begin{figure}[t]
	\centering
    \resizebox{.4\textwidth}{!}{
        \begin{tikzpicture}
    \node[vertex, label=left:$u$] (u) at (0, 0) {};
    \node[vertex, label=below:$v_0$] (v0) at (1, -0.75) {};
    \node[vertex, label=above:$v_1$] (v1) at (1, 0.75) {};
    \node[vertex, label=below:$v_2$] (v2) at (2.5, -0.75) {};
    \node[vertex, label=above:$v_3$] (v3) at (2.5, 0.75) {};
    \node[vertex, label=below:$v_4$] (v4) at (4, -0.75) {};
    \node[vertex, label=above:$v_5$] (v5) at (4, 0.75) {};
    \node[vertex, label=right:$w$] (w) at (5, 0) {};

    \node[vertex] (u0) at (0.5, -0.375) {};
    \node[vertex] (u1) at (0.5, 0.375) {};
    \node[vertex] (u2) at (4.5, -0.375) {};
    \node[vertex] (u3) at (4.5, 0.375) {};

    \draw (u) -- (u0) -- (v0);
    \draw (u) -- (u1) -- (v1);
    \draw (v1) -- (v3);
    \draw (v0) -- (v2);
    \draw[mygreen][thick] (v0) -- (v1);
    \draw[dashed][mygreen][thick] (v2) -- (v3);
    \draw[mygreen][thick] (v4) -- (v5);
    \draw[dashed][myred][thick] (v1) -- (v2);
    \draw[dashed][myred][thick] (v3) -- (v4);
    \draw (v2) -- (v4);
    \draw (v3) -- (v5);
    \draw (v4) -- (u2) -- (w);
    \draw (v5) -- (u3) -- (w);

    \draw[dashed] (u) .. controls +(0,2) and +(0,2) .. (w) node {};
\end{tikzpicture}
    }
	\caption{
        Depicted above is an example of a graph $G$ (solid edges) and augmentation $\widehat{G}$ (dashed edges) that fulfills the conditions of \Cref{theorem:box-inequalities} for $0 \leq x_{uw}$. 
        Essential for the sufficiency proof of this theorem is that the introduced edge sets $H$ and $H_j$ for $j \in \mathbb{N}_0$ have the property $H \subseteq \cup_{j \geq 0} H_j$.
        In the given example, $H = \set{v_0v_1, v_1v_2, v_2v_3, v_3v_4, v_4v_5}$, $H_0 = \set{v_0v_1, v_4v_5}$, $H_1 = \set{v_1v_2, v_3v_4}$, $H_2 = \set{v_2v_3}$, and $H_j = \emptyset$ for $j \geq 3$.
        Thus, $H \subseteq \cup_{j \geq 0} H_j$.
		}
	\label{figure:Hj}
\end{figure}
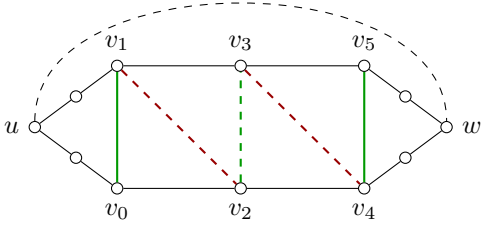

\section{\textsc{NP}-Hardness of Deciding Cut Facets}
\label{section:cut-facets}

\begin{figure}[t]
\centering
\begin{tikzpicture}[scale=0.72]\small
    \node[vertex, label=left:$x_u$] (u) at (0, 0) {};
    \node[vertex, label=above:$\neg\,x_1$] (v1) at (1, 0.75) {};
    \node[vertex, label=below:$x_2$] (v2) at (1, 0) {};
    \node[vertex, label=below:$x_3$] (v3) at (1, -0.75) {};
    \node[vertex, label=below:$x_{d_1}$] (d1) at (2, 0) {};
    \node[vertex, label=below:$x_{d_2}$] (d2) at (3.5, 0) {};
    \node[vertex, label=90:$x_1$] (v4) at (4.5, 0.75) {};
    \node[vertex, label=below:$\neg\,x_1$] (v5) at (4.5, -0.75) {};
    \node[vertex, label=above:$x_2$] (v6) at (5.5, 0.75) {};
    \node[vertex, label=below:$\neg\,x_2$] (v7) at (5.5, -0.75) {};
    \node[vertex, label=above:$x_3$] (v8) at (6.5, 0.75) {};
    \node[vertex, label=below:$\neg\,x_3$] (v9) at (6.5, -0.75) {};
    \node[vertex, label=below:$x_w$] (w) at (7.5, 0) {};
    \node[vertex, label=right:$x_{w'}$] (v_prime) at (8.5, 0) {};

    \draw (u) -- (v1) -- (d1);
    \draw (u) -- (v2) -- (d1);
    \draw[mygreen][thick] (u) -- (v3) -- (d1);

    \draw (d2) -- (v4);
    \draw (v5) -- (d2);
    \draw (v4) -- (v6) -- (v5) -- (v7) -- (v4);
    \draw (v6) -- (v8) -- (v7) -- (v9) -- (v6);
    \draw (v8) -- (w) -- (v9);
    \draw[mygreen][thick] (d2) -- (v4) -- (v7) -- (v8) -- (w);

    \draw[style=cut-edge][myred][thick] (d1) to node[below] {$d$} (d2);
    \draw[style=cut-edge][myred][thick] (v1)[bend left] to node {} (v4);
    \draw[style=cut-edge][myred][thick] (v2)[bend right=45] to node {} (v7);
    \draw[style=cut-edge][myred][thick] (v3)[bend right=45] to node {} (v9);
    
    \draw[style=cut-edge][myred][thick] (d1)[out=90, in=75] to (v_prime);
    \draw(v5)[out=-45, in=-90] to (v_prime);
    \draw(v7)[out=-45, in=-108] to (v_prime);
    \draw(v9)[out=0, in=-135] to (v_prime);
    \draw(v4)[out=45, in=90] to (v_prime);
    \draw(v6)[out=50, in=108] to (v_prime);
    \draw(v8)[out=0, in=135] to (v_prime);
    \draw(w) -- (v_prime);

    \draw[dashed][myred][thick] (u)[bend left=90] to node[pos=0.42,above] {$f$} (w);
\end{tikzpicture}
\caption{Depicted above is an example of the reduction from \textsc{3-sat} used in the proof of \Cref{lemma-complexity-fd}. 
Graphs $G$ and $\widehat{G}$ are constructed from the instance of the \textsc{3-sat} problem given by $\neg \, x_1 \lor x_2 \lor x_3$. 
The additional edge $f$ as well as the edges in the $f$-cut $\delta$ are depicted in red.
The $f_d$-path with respect $\delta$, given by the green edges and $d$, corresponds to the solution of the \textsc{3-sat} problem instance: $\varphi(x_1) = \textsc{false}, \, \varphi(x_2) = \textsc{false} \ \text{and} \ \varphi(x_3) = \textsc{true}$.}
\label{fig:example-np-hardness}
\end{figure}
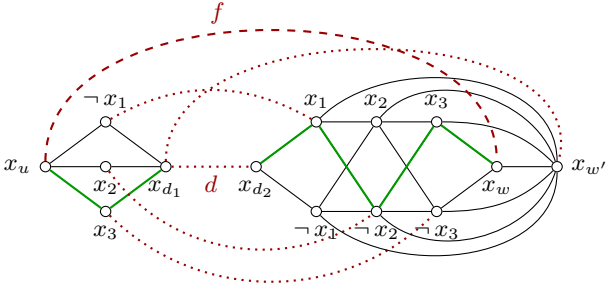

In this section, we prove that it is \textsc{np}-hard to decide facet-definingness of cut inequalities \eqref{eq:cut} for lifted multicut polytopes.
We do so in two steps:
Firstly, we establish a necessary and sufficient condition for facet-definingness of cut inequalities for lifted multicut polytopes in the special case $\lvert F \rvert = 1$ (\Cref{lemma-sufficiency-fd}).
Secondly, we show that deciding this condition for these specific lifted multicut polytopes is \textsc{np}-hard (\Cref{lemma-complexity-fd}).
Together, this implies that facet-definingness is \textsc{np}-hard to decide for cut inequalities of general lifted multicut polytopes (\Cref{theorem:np-hardness}).

We begin by introducing a structure fundamental to this discussion, paths crossing a cut in precisely one edge that have no other edge of the cut as chord:
\begin{definition}
	For 
	any connected graph $G=(V,E)$, 
	any augmentation $\widehat{G} = (V, E \cup F)$ with $F \cap E = \emptyset$, 
	any $f \in F$,
	any $f$-cut $\delta$ of $G$ and
	any $d \in \delta$,
	we call an $f$-path $(V_P, E_P)$ in $G$ an $f_d$-path in $G$ with respect to $\delta$ if and only if it holds for all $d' \in \delta \setminus \set{d}$: $d' \not\subseteq V_P$.
\end{definition}
We proceed by stating the two lemmata and the theorem in terms of $f_d$-paths. 
\begin{lemma}\label{lemma-sufficiency-fd}
	For any connected graph $G=(V,E)$,
	any augmentation $\widehat{G} = (V, E \cup F)$ with $F \cap E = \emptyset$, 
	any $f \in F$ and
	any $f$-cut $\delta$ of $G$,
	it is necessary for the cut inequality $1-x_f \leq \sum_{e \in \delta} (1-x_e)$ to be facet-defining for $\Xi_{G \widehat{G}}$ that an $f_d$-path in $G$ with respect to $\delta$ exists for all $d \in \delta$.
	For the special case of $F = \set{f}$, this condition is also sufficient.
\end{lemma}
\begin{lemma}\label{lemma-complexity-fd}
	For any connected graph $G=(V,E)$,
	any augmentation $\widehat{G} = (V, E \cup F)$ with $F \cap E = \emptyset$,
	any $f \in F$ and
	any $f$-cut $\delta$ of $G$,
	it is \textsc{np}-hard to decide if an $f_d$-path in $G$ with respect to $\delta$ exists for all $d \in \delta$, even for the special case of $F = \set{f}$.
\end{lemma}
\begin{theorem}
\label{theorem:np-hardness}
	For any connected graph $G=(V,E)$,
	any augmentation $\widehat{G} = (V, E \cup F)$ with $F \cap E = \emptyset$,
	any $f \in F$ and any $f$-cut $\delta$ of $G$,
	it is \textsc{np}-hard to decide if the cut inequality $1-x_f \leq \sum_{e \in \delta} (1-x_e)$ is facet-defining for $\Xi_{G \widehat{G}}$, even for the special case of $F = \set{f}$.
\end{theorem}
In the remainder of this section, we prove first \Cref{theorem:np-hardness} and then \Cref{lemma-sufficiency-fd} and \Cref{lemma-complexity-fd}.
\begin{proof}[Proof of \Cref{theorem:np-hardness}]
	In case $F = \set{f}$, a cut inequality is facet defining if and only if there exists an $f_d$-path in $G$ with respect to $\delta$ for all $d \in \delta$, by \Cref{lemma-sufficiency-fd}.
	Deciding if such paths exist is \textsc{np}-hard, by \Cref{lemma-complexity-fd}. 
	Together, this implies \textsc{np}-hardness of deciding facet-definingness, even for the special case of $F = \set{f}$.
\end{proof}

\begin{proof}[Proof of \Cref{lemma-sufficiency-fd}]
	Necessity of an equivalent statement was already proven as Condition $C1$ of Theorem 5 of \citet{andres-2023-a-polyhedral}.

	We now show sufficiency.
	For this, let $F = \{f\} = \{uw\}$, let $\sigma = {\{x \in X_{G \widehat{G}} \ \mid \ 1 - x_f = \sum_{d \in \delta} (1-x_d)\}}$ and assume that there exists an $f_d$-path in $G$ with respect to $\delta$ for all $d \in \delta$.
	We prove that the cut inequality with respect to $f$ and $\delta$ is facet-defining under the specified conditions by explicitly constructing ${\lvert E \cup F \rvert - 1} = \lvert E \rvert$ linearly independent vectors in the linear space $\lin(\set{x-y \mid x,y \in \sigma})$ which we abbreviate by ${\lin(\sigma - \sigma)}$, implying $\dim \aff \sigma = \dim \lin(\sigma - \sigma) = \lvert E \rvert$. In particular, we construct $\mathbbm{1}_\set{d,f}$ for all $d \in \delta$ and the characteristic vectors of the elements in $E \setminus \delta$.

	For any $e \in E \setminus \delta$, define:
	\begin{equation*}
		V_1 = \set{e} \quad V_2 = \emptyset
		\enspace .
	\end{equation*}
	As $G[U]$ is connected for any $U \in V_1$ and $U \in V_2$, we have $x^{V_1}, x^{V_2} \in X_{G \widehat{G}}$, by \Cref{lemma:construction}.
	It further holds $\mathbbm{1}_\set{e} = -x^{V_1}+x^{V_2}$ and, for $j \in \{0,1\}$, that $x^{V_j}_{uw} = 1$ and $x^{V_j}_{d} = 1$ for all $d \in \delta$, as for all $pq \in E \cup F$:
	\begin{itemize}
		\item $x^{V_1}_{pq} = 1 \ \text{and} \ x^{V_2}_{pq} = 1 \ \text{if} \ \nexists U \in V_1\colon \set{p,q} \subseteq U$
		\item $x^{V_1}_{pq} = 0 \ \text{and} \ x^{V_2}_{pq} = 1 \ \text{if} \ \exists U \in V_1\colon \set{p,q} \subseteq U$\,.
	\end{itemize}
	It follows from $x^{V_j}_{uw} = 1$ and $x^{V_j}_{d} = 1$ for all $d \in \delta$ that $x^{V_j} \in \sigma$. 
	Thus, $\mathbbm{1}_\set{e} = -x^{V_1}+x^{V_2} \in \lin(\sigma - \sigma)$, which concludes the first case.

	For any $d \in \delta$, there exists an $f_d$-path $P=(V_{P}, E_{P})$ in $G$ with respect to $\delta$ according to our assumptions.
	We define:
	\begin{equation*}
		V_1 = \set{V_{P}} \quad V_2 = \emptyset 
		\enspace .
	\end{equation*}
	Analogously to the previous case, we get $x^{V_1} \in X_{G \widehat{G}}$, $x^{V_2} \in \sigma$ and $\mathbbm{1}_{E_{P} \cup \set{f}} = -x^{V_1}+x^{V_2}$.
	Using the same distinction of cases as before, we further get $x^{V_1}_{uw} = 0$ and, as $P$ is an $f_d$ path, $x^{V_1}_d = 0$ and $x^{V_1}_{d'} = 1$ for all $d' \in \delta \setminus \set{d}$, implying $x^{V_1} \in \sigma$.
	Consequently, $\mathbbm{1}_{E_{P} \cup \set{f}} = -x^{V_1}+x^{V_2} \in \lin(\sigma - \sigma)$.
	We now note that the characteristic vector associated with $f$ and $d$ can be written as
	\begin{equation}
		\mathbbm{1}_{\set{f,d}} = \mathbbm{1}_{E_{P} \cup \set{f}} - \sum_{e \in E_{P} \setminus \set{d}} \mathbbm{1}_{\set{e}} \enspace .
	\end{equation}
	As $\mathbbm{1}_{\set{e}} \in \lin(\sigma - \sigma)$ for all $e \in E_{P} \setminus \set{d}$ by the previous case, this implies $\mathbbm{1}_{\set{f,d}} \in \lin(\sigma - \sigma)$ and concludes the second case.
	Altogether, we have constructed $\lvert E \rvert$ linearly independent vectors in $\lin(\sigma - \sigma)$ and have thus established sufficiency of the specified condition.
\end{proof}

\begin{proof}[Proof of \Cref{lemma-complexity-fd}]
	For showing \textsc{np}-hardness, we use a reduction from the \textsc{np}-hard \textsc{3-sat} problem with exactly three literals per clause and no duplicating literals within clauses \citep{3Sat}.
	For any instance of this \textsc{3-sat} problem, with variables $x_1,x_2, \ldots , x_n$ and clauses $C_1,C_2, \ldots , C_m$, we construct in polynomial time an instance of our decision problem and show that it has a solution if and only if the instance of the \textsc{3-sat} problem has a solution.
	An example of this construction is depicted in \Cref{fig:example-np-hardness}.
	We begin by defining two graphs, $G_1$ and $G_2$, which will be the components of $G$ induced by the $f$-cut $\delta$ of our original decision problem.

	In the first graph $G_1=(V_1, E_1)$, there are $3m+2$ nodes which are organized in $m+2$ fully-connected layers. 
	For $j \in \set{0,1,\ldots,m+1}$, we denote the set of nodes in the $j$-th layer by $V_{1j}$.
	The $0$-th layer contains a single node $u$ and the $m+1$-th layer a single node $d_1$. 
	The remaining $m$ layers correspond to the $m$ clauses $C_1, C_2, \ldots , C_m$ and contain three nodes each. 
	The edges between consecutive layers are the only edges in $E_1$.
	For $j \in \set{1,2,\ldots,m}$, we label each node in the $j$-th layer by a different literal in $C_j$.
	For completeness, we label $u$ (respectively $d_1$) by a unique auxiliary propositional variable $x_u$ (respectively $x_{d_1}$).
	For any $v \in V_1$, we let $l(v)$ denote the label of that node.

	The second graph $G_2=(V_2, E_2)$ is such that $V_1 \cap V_2 = \emptyset$ and $E_1 \cap E_2 = \emptyset$.
	It consists of $2n+3$ nodes which are organized in $n+3$ fully-connected layers.
	For $k \in \set{0,1,\ldots,n+2}$, we denote the set of nodes in the $k$-th layer by $V_{2k}$. 
	The $0$-th layer contains a single node $d_2$, the $n+1$-th layer a single node $w$ and the $n+2$-th layer a single node $w^\prime$, which is connected to all other nodes of $G_2$, besides $d_2$, by a set of edges $E_2^\prime \subseteq E_2$. 
	The remaining $n$ layers correspond to the $n$ variables $x_1,x_2, \ldots , x_n$ and contain two nodes each.
	The edges between consecutive layers and the edges in $E'_2$ are the only edges in $E_2$.
	For $k \in \set{1,2,\ldots,n}$, we label one node in the $k$-th layer by $x_k$ and the other by $\neg\,x_k$.
    Again, we label $w$ (respectively $d_2$ and $w'$) by a unique auxiliary propositional variable $x_w$ (respectively $x_{d_2}$ and $x_{w'}$) and denote the label of any $v \in V_2$ by $l(v)$.

	We construct a third graph $G=(V,E)$ such that $V = V_1 \cup V_2$ and $E = E_1 \cup E_2 \cup \delta$ with
	\begin{align*}
		\delta =\;
        & \bigl\{d_1d_2, d_1w'\bigr\} \ \cup \\
        & \bigl\{
            st \subseteq V_1 \cup V_2 \mid\, s \in V_1 \wedge t \in V_2 \wedge l(s) = \neg\, l(t)         
          \bigr\}
          \enspace .
	\end{align*}
	For brevity, we introduce the symbol $d := d_1 d_2$.
	Finally, we define a fourth graph $\widehat{G} = (V, E \cup F)$ such that $F = \set{f} = \set{uw}$.
    Note that $G$ is connected and that $\delta$ is an $f$-cut of $G$, partitioning it into $V_1$ and $V_2$.
	Note also that $|F| = 1$, covering the part of the lemma claiming \textsc{np}-hardness also for this special case. 
	
	Henceforth, we mean by an $f_d$-path an $f_d$-path in $G$ with respect to $\delta$.
	\begin{claim}\label{claim-graph-properties}
		The graph $G$ has the following properties:
		\begin{enumerate}[label=\textnormal{(\roman*)}]
			
			\item \label{item-graph-property-2} For any clause $C_j$ and any $f_d$-path $(V_P, E_P)$, there exists a literal in $C_j$ labeled by a node from $V_P \cap V_1$. 
			\item \label{item-graph-property-3} Any $f_d$-path that contains a node in $V_1$ labeled $\neg\,x_k$ (respectively $x_k$) does not contain a node labeled $x_k$ (respectively $\neg\,x_k$).
			
		\end{enumerate}
	\end{claim}
	Using \Cref{claim-graph-properties}, which is proven in the appendix, we show that the \textsc{3-sat} formula is satisfiable if and only if there exists an $f_{d^\prime}$-path for every $d^\prime \in \delta$.
	We do so in two steps:
	Firstly, we show that the \textsc{3-sat} formula is satisfiable if and only if there exists an $f_d$-path for the specific edge $d \in \delta$.
	Secondly, we show that there always exists an $f_{d^\prime}$-path for every other edge $d^\prime \in \delta \setminus \set{d}$.

	Let $P=(V_P, E_P)$ be an $f_d$-path.
	We construct an assignment of truth values $\varphi$ to the variables $x_1,x_2, \ldots, x_n$ satisfying the corresponding \textsc{3-sat} problem instance by setting $\varphi(x_k) = \textsc{true}$ for all $k \in \set{1, \ldots, n}$ if and only if there exists a node $v \in V_P \cap V_1$ such that $l(v) = x_k$.
	Assume this assignment would not satisfy the \textsc{3-sat} problem instance. Then there exists a clause $C_j$ assigning \textsc{false} to all of its labels.
	By \ref{item-graph-property-2}, there exists a node $v \in V_P \cap V_1$ that is labeled by a literal in $C_j$.
	If $l(v) = x_k$ for some variable $x_k$, then $\varphi(x_k) = \textsc{true}$, leading $C_j$ to be true. 
	If $l(v) = \neg\,x_k$, then $\varphi(x_k) = \textsc{false}$ by \ref{item-graph-property-3}, leading $C_j$ to be true as well.
	Consequently, such a clause $C_j$ where all literals get assigned \textsc{false} cannot exist and $\varphi$ is a solution to the given \textsc{3-sat} problem instance.

	Let now $\varphi$ be an assignment of truth values to the variables $x_1,x_2, \ldots, x_n$ that satisfies the corresponding instance of the \textsc{3-sat} problem.
	In the following, we will show that an $f_d$-path $P = (V_P, E_P)$ in $G$ is given by
	\begin{align*}
		V_P &= 
			\bigl\{
				u, u_1,\ldots, u_m, d_1, d_2, w_1, \ldots, w_n, w
			\bigr\} \\
		E_P &= 
			\bigl\{
				uu_1, u_1u_2, \ldots, u_md_1, d_1d_2, \\ 
				& \hspace{5ex} d_2w_1, w_1w_2, \ldots, w_nw
			\bigr\}
		\enspace,
	\end{align*}
	where $u_j \in V_{1j}$ (respectively $w_k \in V_{2k}$) has a label that gets assigned $\textsc{true}$ by $\varphi$ for all $j \in \set{1, \ldots, m}$ (respectively $k \in \set{1, \ldots, n}$).
	It is easy to see that $P$ is an $f$-path in $G$.
	It remains to show that it is an $f_d$-path, i.e. that there exist no $d^* = d^*_1d^*_2 \in \delta \setminus \set{d}$ such that $d^* \subseteq V_P$.
	Assume there exists such a $d^*$.
    As $w' \not\in V_P$, it holds then $d^* \in \delta \setminus \set{d,d_1w'}$.
	By construction of $\delta$, it follows $l(d^*_1) = \neg\, l(d^*_2)$. 
	As both $l(d^*_1)$ and $\neg\, l(d^*_2)$ need to get assigned \textsc{true} by $\varphi$ according to the construction of $P$, this is a contradiction.
	Thus, there exists no such $d^*$, and $P$ is an $f_d$-path.
	For an example of this correspondence between $f_d$-paths and solutions of the given \textsc{3-sat} problem instance, see again \Cref{fig:example-np-hardness}.

	Next, we regard the other edges of the cut.
	Let $d' = d'_1d'_2 \in \delta \setminus \set{d}$ be an edge in the cut except $d$.
	We assume w.l.o.g. that $d'_1 \in V_{1i}$ for some $i \in \set{1, \ldots, m+1}$ and regard the path $P = (V_P, E_P)$ given by
	\begin{align*}
		V_P &= 
		\bigl\{
			u, u_1, u_2, \ldots, u_{i-1}, d'_1, d'_2, w', w
		\bigr\} \\
		E_P &= 
		\bigl\{
			uu_1, u_1u_2, \ldots, u_{i-1}d'_1, d^\prime_1d^\prime_2, d^\prime_2w', w'w
		\bigr\}
		\enspace,
	\end{align*}
	where $u_j$ is an arbitrary node in $V_{1j}$ such that $l(u_j) \neq l(d'_1)$ for all $j \in \set{1, \ldots i-1}$.
	Note that such $u_j$ are guaranteed to exist as we consider the \textsc{3-sat} problem with exactly three literals per clause and no duplicated literals within clauses. 
	Again, it is easy to see that $P$ is an $f$-path, and it remains to show that there exists no $d^* = d^*_1d^*_2 \in \delta \setminus \set{d'}$ such that $d^* \subseteq V_P$.
	Assume there exists such a $d^*$.
	Then one if its nodes, say $d^*_1$, must be in $V_1$ and its other node must be in $V_2$.
	We make a case distinction on whether $d^*_1 \in V_1 \setminus \set{d_1}$ or $d^*_1 = d_1$.
	If $d^*_1 \in V_1 \setminus \set{d_1}$, then $d^* \in \delta \setminus \set{d,d_1w'}$.
	By construction of $\delta$, it follows $l(d^*_1) = \neg\, l(d^*_2)$.
	As $d^*_2 \in V_2 \cap V_P = \set{d'_2, w', w}$ and $l(d^*_1) \neq l(d'_1)$ according to the construction of $P$, this is a contradiction.
	On the other hand, if $d^*_1 = d_1$, it holds by construction of $P$ that $i=m+1$.
	As $d'_1 \in V_{1i} = V_{1m+1} = \set{d_1}$, it follows $d'_1 = d_1 = d^*_1$.
	Furthermore, as $d_1d_2$ and $d_1w'$ are the only edges in $\delta$ containing $d_1$ and $d' \neq d$, we get $d'_2 = w'$.
	Thus, we especially have $d_2 \not\in V_2 \cap V_P = \set{w', w}$, leading to $d^*_2 = w'$ when using the same argument as before.
	Consequently, $d^* = d'$ which contradicts $d^* \in \delta \setminus \set{d'}$.
	As both cases lead to a contradiction, there does not exist such a $d^*$ and $P$ is an $f_d$-path.
    This finishes the reduction from the \textsc{3-sat} problem and the proof of the lemma.
\end{proof}

\section{Conclusion}
\label{section:conclusion}

We characterize in terms of efficiently decidable conditions the facets of lifted multicut polytopes induced by lower box inequalities.
In addition, we show that deciding facet-definingness of cut inequalities for lifted multicut polytopes is \textsc{np}-hard.
Toward the design of cutting plane algorithms for the lifted multicut problem, our hardness result does not rule out the existence of inequalities strengthening the cut inequalities for which facet-definingness and possibly also the separation problem can be solved efficiently.
The search for such inequalities is one direction of future work.
In our proof, we identify a structure (paths crossing the cut that have an edge of the cut as a chord) that complicates the characterization of cut inequalities.
This structure exists for cuts (edge subsets, as discussed in this article) but does not exist for separators (node subsets, not discussed in this article).
This observation motivates the study of non-local connectedness with respect to separators instead of cuts.
\paragraph{Acknowledgements} This work is partly supported by the Federal Ministry of Education and Research of Germany through
DAAD Project 57616814 (\href{https://secai.org/}{SECAI}).

\appendix
\section{Additional Proofs}

\begin{claimproof}[Proof of \Cref{claim:H-partition}]
	It follows directly from \eqref{equation:definition-Hi} that $H_j \cap H_k = \emptyset$ for any distinct $j,k \in \mathbb{N}_0$ and that $\cup_{j \geq 0} H_j \subseteq H$.
	Let $H_\infty = H \setminus \cup_{j \geq 0} H_j$, it remains to show that $H_\infty = \emptyset$.
	Assume this does not hold, then there exists an $st \in H_\infty$, and thus especially a simple $st$-path $(\set{s,t}, \set{st})$ in $\widehat{G}$ whose edges are all in $H_\infty$.
	We show that such a path cannot exist given \ref{theorem:condition:path} and \ref{theorem:condition:cycle}.

	Assume there exist simple paths in $\widehat{G}$ whose edges are all in $H_\infty$.
	As the number of edges in $\widehat{G}$ is finite, there exists a maximum length of such paths.
	Let $P = (V_P, E_P)$ with $E_P \subseteq H_\infty$ be one of those simple paths with maximum length, let $p,q \in V$ be its end-nodes and let $e_p, e_q \in E_P$ be the unique edges in $E_P$ containing $p$ and $q$, respectively.
	Recall that, by definition of $H$, all edges in $E_P$ are $uw$-separators of $G$ and no node in $V_P$ is a $uw$-cut-node of $G$.
	By \eqref{equation:definition-Hi}, there exists a $qq' \in E_{\widehat{G}}(q) \setminus \set{e_q}$ such that $qq'$ is a $uw$-separator and $qq' \not\in \cup_{j \geq 0} H_j$.
	By definition of $H$, this is equivalent to $q'$ being either a $uw$-cut-node of $G$ or $qq' \in H_\infty$.
	It is not possible that $qq' \in H_\infty$, as either $q' \in V_P$ and there exists a cycle in $(V_P, E_P \cup \set{qq'})$ that violates \ref{theorem:condition:cycle}, or $q' \not\in V_P$ and $(V_P \cup \set{q'}, E_P \cup \set{qq'})$ is a simple path in $G$ whose edges are all in $H_\infty$, contradicting $P$ to be the longest such path.
	Thus, $q'$ must be a $uw$-cut-node.
	Further, it holds $q' \not\in V_P$ as no node in $V_P$ is a $uw$-cut-node of $G$.
	Analogously, there must exist a $uw$-cut-node $p' \in V \setminus V_P$ such that $p'p \in E_{\widehat{G}}(p) \setminus \set{e_p}$.

	If $p' = q'$, the simple cycle $(V_{P} \cup \set{p'}, E_{P} \cup \set{p'p, qq'})$ violates \ref{theorem:condition:cycle}.
	If $p' \neq q'$, the simple $p'q'$-path $(V_{P} \cup \set{p',q'}, E_{P} \cup \set{p'p, qq'})$ violates \ref{theorem:condition:path}.
	As both cases lead to a contradiction, there cannot exist simple paths in $\widehat{G}$ whose edges are all in $H_\infty$, and thus especially no $st \in H_\infty$.
\end{claimproof}

\begin{claimproof}[Proof of \Cref{claim-graph-properties}]
	For proving \ref{item-graph-property-2} and \ref{item-graph-property-3}, we first show that any $f_d$-path contains one node from each layer of $G$ besides $V_{2n+2} = \set{w^\prime}$. 
	Assume this does not hold.
	Then there exists an $f_d$-path $P=(V_P, E_P)$ and a layer $V_{1j}$ with $j \in \set{0,\ldots,m+1}$ or $V_{2k}$ with $k \in \set{0,\ldots,n+1}$ such that no node in this layer is contained in $P$. 
	As $P$ is an $f_d$-path, it holds that $d_1 \in V_P$ and $E_P \cap \delta = \set{d}$.
	As $d_1w' \in \delta$, this especially implies that $w' \not\in V_P$ and thus $E_P \cap E_2^\prime = \emptyset$.
	Hence, $E_C$ must be a subset of the remaining edges $E_1 \cup E_2 \cup \set{d} \setminus E_2^\prime$.
	As these edges only exist between consecutive layers and $P$ contains $u \in V_{1,0}$ and $w \in V_{2,n+1}$, having a layer in-between for which $P$ does not contain a node would imply $P$ not being connected and thus results in a contradiction.

	Assume \ref{item-graph-property-2} does not hold.
	Then there exists a clause $C_j$ and an $f_d$-path $P$ such that no node in $V_P \cap V_1$ is labeled by a literal in $C_j$.
	By construction of the labels, this would imply that there exists no node in $P$ that is in $V_{1j}$, contradicting the discussion of the previous paragraph.

	Assume \ref{item-graph-property-3} does not hold.
	Then there exists an $f_d$-path $P=(V_P, E_P)$ containing an $s \in V_1 \cap V_P$ with $l(s) = \neg\,x_k$ (respectively $x_k$) and a $t \in V_P$ with $l(t) = x_k$ (respectively $\neg\,x_k$). 
	We make a case distinction depending on whether $t$ is in $V_1$ or $V_2$.
	Suppose $t \in V_1$.
	By the discussion of the first paragraph, $P$ contains a $v \in V_{2k} \cap V_P$ with either $l(v) = x_k$ or $l(v) = \neg\,x_k$.
	By construction of $\delta$, it follows either $sv \in \delta \setminus \set{d}$ or $tv \in \delta \setminus \set{d}$, contradicting $P$ to be an $f_d$-path.
	Suppose $t \in V_2$.
	In this case, $st \in \delta \setminus \set{d}$ by construction of $\delta$, contradicting $P$ to be an $f_d$-path.
	As both cases lead a contradiction, such nodes $s$ and $t$ cannot exist.
\end{claimproof}

\bibliography{main}
\bibliographystyle{plainnat}

\end{document}